\documentclass{rspublic}
\usepackage{amsfonts,amssymb}
\usepackage{epsfig}
\usepackage{url}

\begin{document}

\title[ Riemann Hypothesis for Dirichlet $L$ Functions]{The Riemann Hypothesis for Dirichlet $L$ Functions}
\author[R.C. McPhedran ]{Ross C. McPhedran$^{1,2}$
}

\affiliation{$^1$CUDOS, School of Physics, University of Sydney, NSW 2006, Australia\\
$^2$ Department of Mathematical Sciences, University of Liverpool, Liverpool L69 7ZL, United Kingdom.
}
\label{firstpage}

\maketitle
\begin{abstract}{Lattice sums, Riemann hypothesis, Dirichlet $L$ functions, distribution functions of zeros}
This paper studies the connections between the zeros and their distribution functions for two particular Dirichlet $L$ functions: the Riemann zeta function, and the Catalan beta function, also known as the Dirichlet beta function.
It is shown that the Riemann hypothesis holds for the Dirichlet beta function $L_{-4}(s)$ if and only if it holds for $\zeta (s)$- a particular case of the Generalized Riemann Hypothesis.
\end{abstract}
\section{Introduction}
This paper continues  the investigations into the Riemann hypothesis for sums related to the Riemann zeta function. In a previous paper ( McPhedran et al 2011, hereafter referred to as I)
it was established that the Riemann hypothesis held for all of a particular class of two-dimensional sums over the square lattice if and only if it held for the lowest member of that class. That lowest
member, denoted ${\cal C}(0,1;s)$ was known from  the work of Lorenz (1871) and Hardy (1920)
to be given by the product of the Riemann zeta function $\zeta (s)$, and the particular Dirichlet $L$  or beta function, $L_{-4}(s)$ (Zucker and Robertson, 1976). 
In I it was also established that, if the complex variable $s$ is written as
$\sigma+ i t$, and all zeros of ${\cal C}(0,1;s)$ lie on $\sigma=1/2$ by a generalisation of the Riemann hypothesis, then the whole class of sums has the same distribution function for zeros,
as far as all terms tending to infinity with $t$ are concerned.  In another paper (McPhedran \& Poulton, 2014,
hereafter referred to as II) it was established that an antisymmetric and a symmetric combination of two Riemann zeta functions obeyed the Riemann hypothesis, with the result for the antisymmetric combination being an alternative proof to one given by Taylor (1945).
Numerical evidence was given that the two combinations of zeta functions both have the same distribution functions of zeros on  the critical line $\sigma=1/2$, and that this is also the same as that of ${\cal C}(0,1;s)$ and $\zeta (2 s-1/2)$.

Here, we will investigate the connection between the distributions  and locations of zeros
of $\zeta (s)$, $L_{-4}(s)$ and $\zeta (2 s-1/2)$. We will do this by studying the properties of the
quotient function $\zeta(s) L_{-4}(s)/\zeta(2 s-1/2)$, denoted as $\Delta_5(s)$, in a similar way to the course followed in I. We will prove that the Riemann hypothesis holds for $L_{-4}(s)$ if and only if it holds for $\zeta (s)$- a particular case of the Generalized Riemann Hypothesis. We will also prove the corollary that, if the Riemann hypothesis holds for $\zeta(s)$ in the range $0<t<t_0$, then it also holds  for $L_{-4}(s)$ at least  in the range $0<t<t_0/2$.

\section{Properties of $\Delta_5(s)$}
\begin{theorem}
The analytic function $\Delta_5(s)$, defined by
\begin{equation}
\Delta_5(s)=\frac{\zeta(s) L_{-4}(s)}{\zeta(2 s-1/2)}
\label{defn}
\end{equation}
 obeys the functional equation
\begin{equation}
\Delta_5(s)={\cal F}_5(s) \Delta_5(1-s),
\label{del5-1}
\end{equation} 
where
\begin{equation}
{\cal F}_5(s) =\frac{\Gamma (1-s) \Gamma(s-\frac{1}{4})}{\Gamma (s) \Gamma(\frac{3}{4}-s)}.
\label{del5-2}
\end{equation}
\label{thm5-1}
It is monotonic decreasing along the real axis, apart from jumps at first-order poles. The poles
are at $\sigma=1,-3/4,-7/4,-11/4, \ldots$, and are separated by first order zeros at
$\sigma=3/4, -1,-2,-3, \ldots$. Its phase on the critical line (modulo $\pi$) is that of $\Gamma(1/2- it) \Gamma(1/4 +it)$, which is for $t$ not small well approximated by $-\pi/8-1/(32 t)$, modulo $\pi$.
\end{theorem}
\begin{proof}
The functional equation for $\Delta_5(s)$ follows from that for ${\cal C}(0,1;s)$ (see I) and
that for the Riemann zeta function (Titchmarsh \& Heath-Brown, 1987):
\begin{equation}
\label{del5-3}
\frac{\Gamma (s) {\cal C}(0,1;s)}{\pi^s}=\frac{\Gamma (1-s) {\cal C}(0,1;1-s)}{\pi^{(1-s)}},~~
\frac{\Gamma (s-1/4) \zeta(2 s-1/2)}{\pi^{s-1/4}}=\frac{\Gamma (3/4-s) \zeta(3/2-2 s)}{\pi^{(3/4-s)}}.
\end{equation}

In $\sigma>1$, $\Delta_5(s)$ can be represented by its series obtained by direct summation,
of which the first few terms are
\begin{equation}
\Delta_5(s)= 1+\frac{1}{2^s}+\frac{(1-\sqrt{2})}{4^s}+\frac{2}{5^s}+\ldots.
\label{del5-4}
\end{equation}
The pole at $\sigma =1$ is that of the function $\zeta(s)$ in the numerator of
$\Delta_5(s)$, and the zero at $s=3/4$ is that of $\zeta (2 s-1/2)$ in the denominator.
The positions of other poles and zeros on the real axis then follow from the functional equation (\ref{del5-1}), together with the properties of the gamma functions in (\ref{del5-2}) . The residues
of $\Delta_5(s)$ at its first few poles are (1, 0.300645), (-3/4, 0.312673), (-7/4, 0.25505), (-11/4, 0.237821), (-15/4, 0.230136) and (-19/4, 0.22657). The coefficients of the linear term round its first few zeros are (3/4, -5.0378),  (-1,-5.7055), (-2, -4.9245), (-3, -4.645) and (-4, -4.51975).

We next expand the expression for ${\cal F}_5(s)$ using Stirling's formula, assuming $|s|$ large:
\begin{equation}
{\cal F}_5(s)=\frac{\sin[\pi (3/4-s)]}{\sin (\pi s)} \left(1+\frac{1}{16 s}+\frac{17}{512 s^2} +\ldots \right).
\label{del5-5}
\end{equation}
We expand the $\sin$ term in the numerator in (\ref{del5-5}), and then replace both $\sin$ and $\cos$ terms by their representations in terms of $\exp ( i\pi s)$ and $\exp (-i \pi s)$. Assuming $t>1$,
it is an accurate approximation to neglect terms in the former exponential in comparison with terms
in the latter exponential. This leads to the approximation in $t>1$
\begin{equation}
{\cal F}_5(s)\simeq e^{-i\pi/4} \left( 1+\frac{1}{16 s}+\frac{17}{512 s^2} +\ldots \right).
\label{del5-6}
\end{equation}
(In $t<-1$, the  factor $e^{-i\pi/4}$ is replaced by $e^{+i\pi/4}$.)

We then use the expression
\begin{equation}
2 \arg \Delta_5 (1/2+ it )= \arg {\cal F}_5(1/2+i t),~ {\rm modulo}~2 \pi
\label{del5-7}
\end{equation}
which follows from the functional equation (\ref{del5-1}), together with (\ref{del5-6}). This yields
\begin{equation}
\arg \Delta_5 (1/2+ it )\simeq  -\frac{\pi}{8}-\frac{1}{32 t}, ~ {\rm modulo}~\pi
\label{del5-8}
\end{equation}
in $ t>1$.
\end{proof}

\begin{figure}[h]
\includegraphics[width=3.0in]{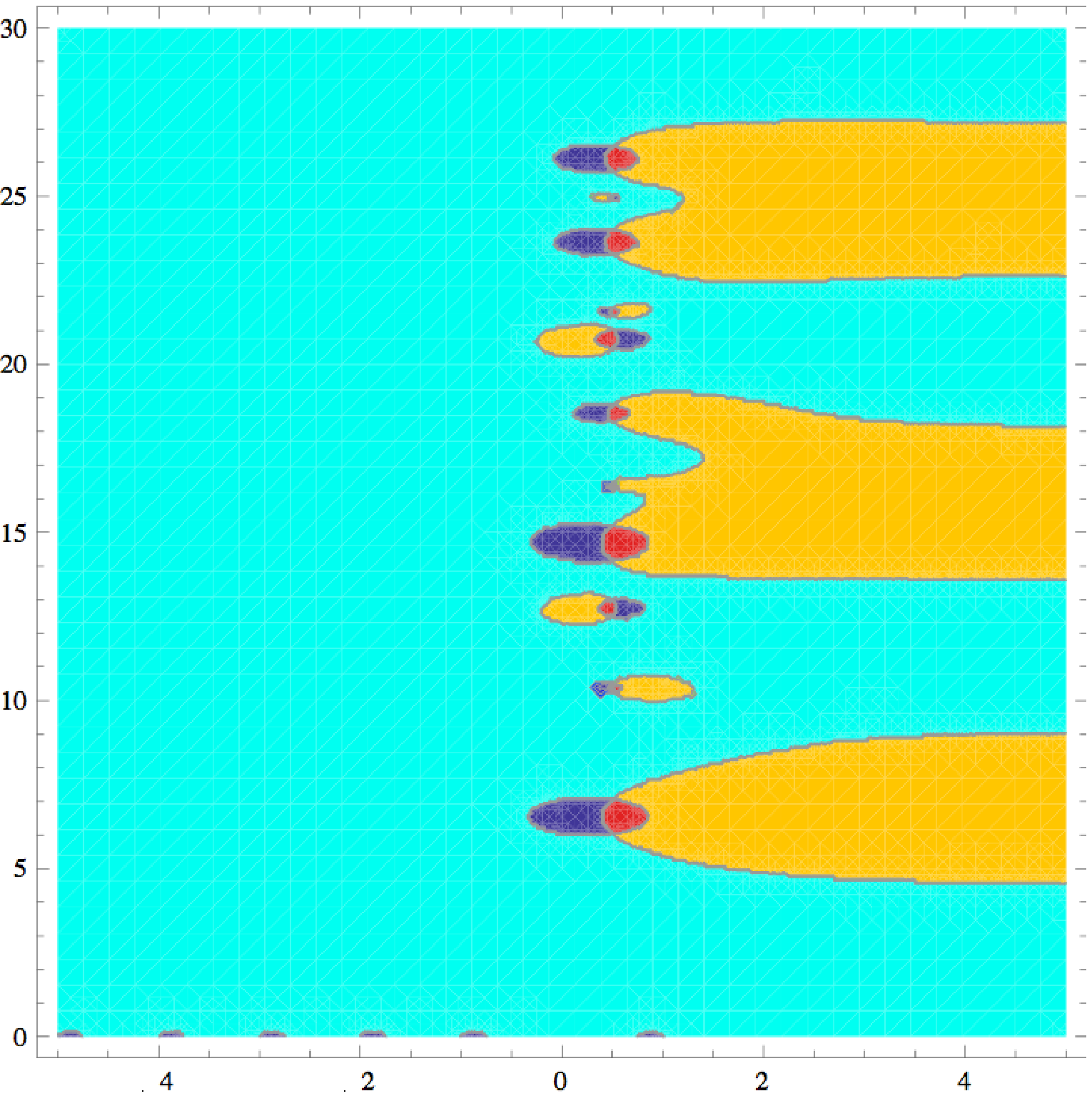}~~\includegraphics[width=3.0in]{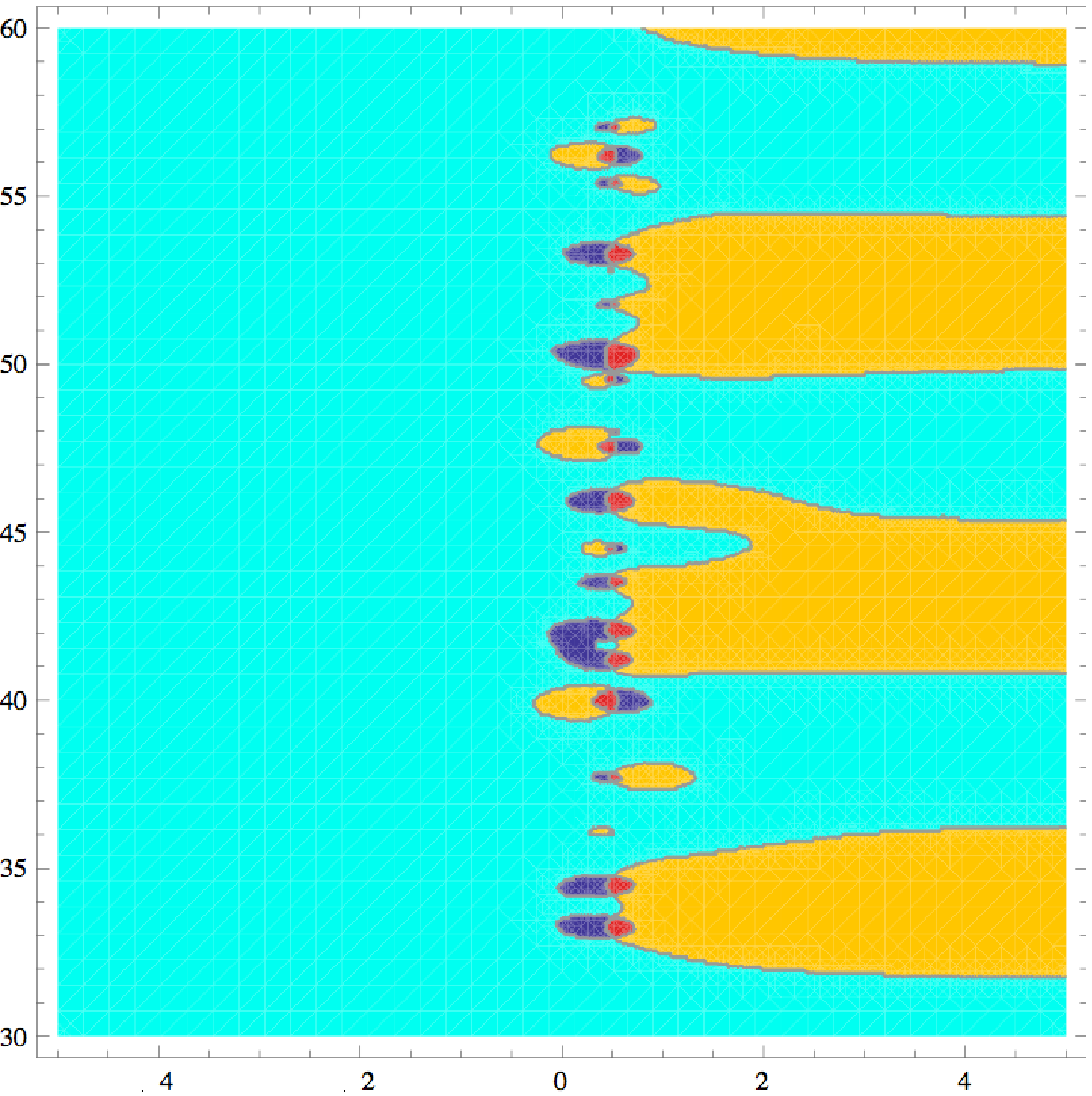}
\caption{Plots of the phase of $\Delta_5(s)$ in the complex plane. The phase is denoted by colour according to which
quadrant it lies in:  yellow-first quadrant,  red- second, purple- third, and light blue- fourth.}
\label{fig01}
 \end{figure}

\begin{figure}[h]
\includegraphics[width=3.0in]{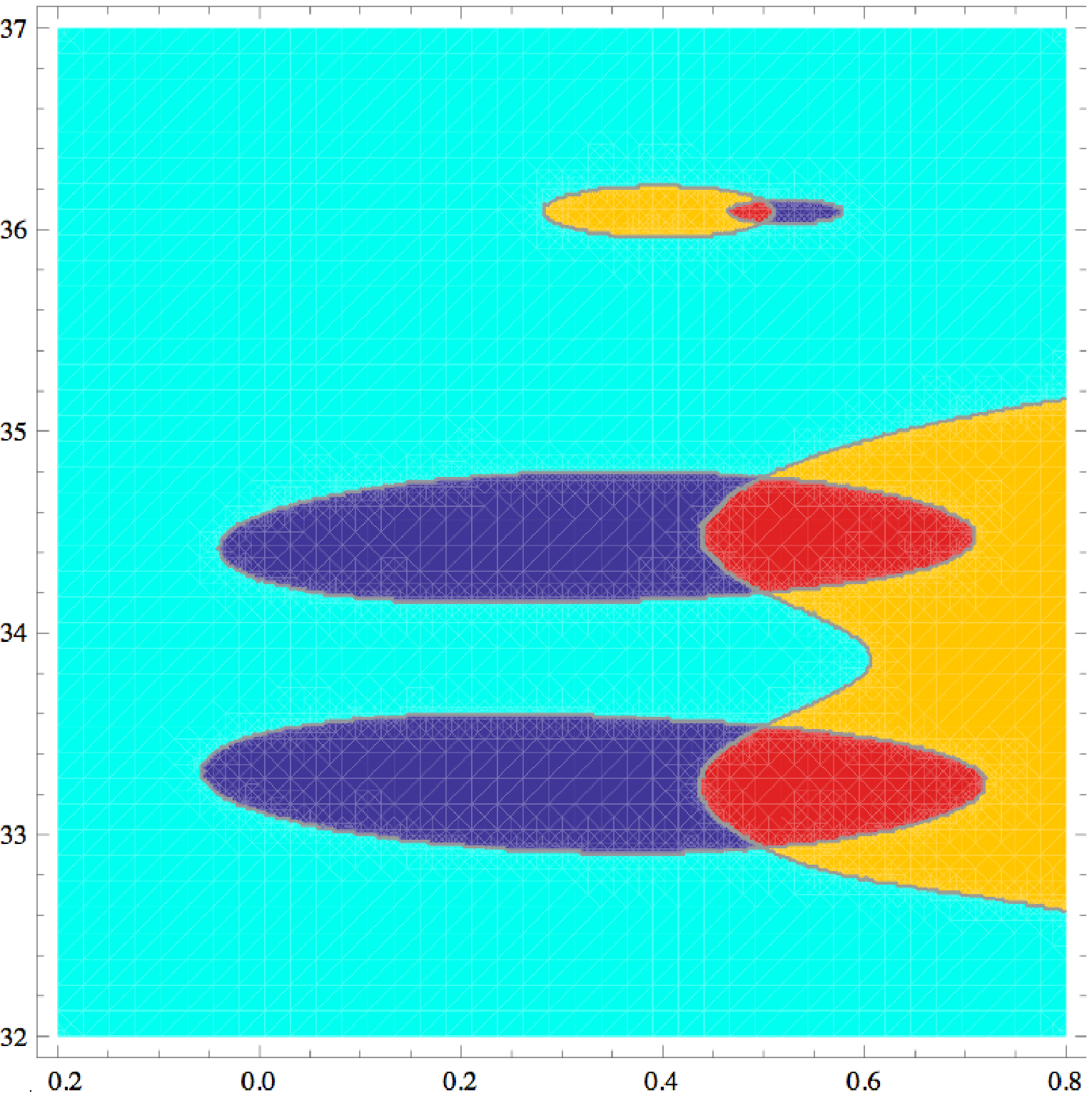}~~\includegraphics[width=3.0in]{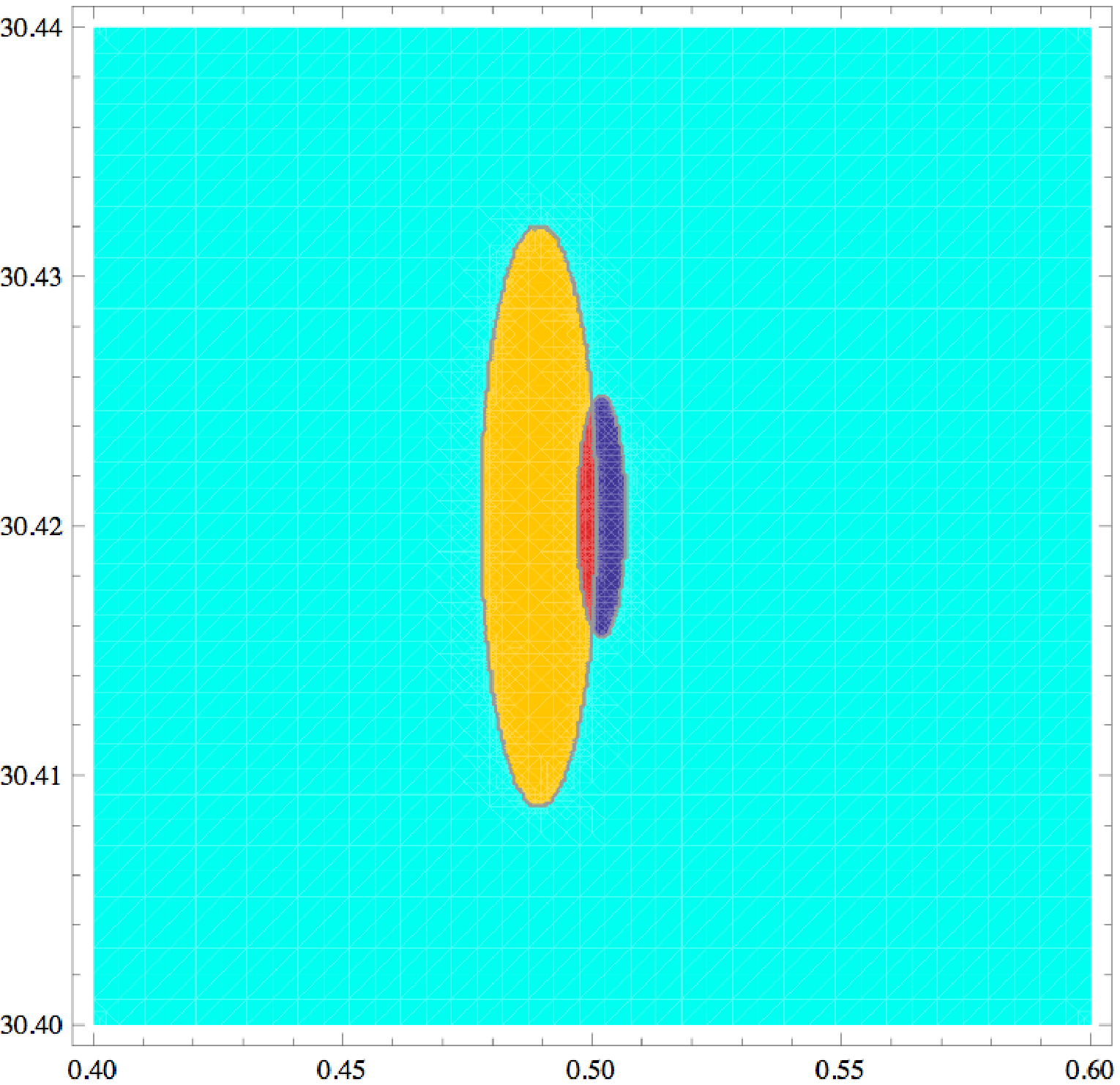}
\caption{ Further detail of the plots of the phase of $\Delta_5(s)$ in the complex plane
in Fig. \ref{fig01}. Zeros (Z) and poles (P) of  $\Delta_5(s)$ occur where four quadrants intersect,
and occur in the left panel in order of increasing $t$: Z, P, Z, P, P, Z and in the right panel P, Z.}
\label{fig02}
 \end{figure}

In Figs. \ref{fig01} and \ref{fig02} we show quadrant phase plots of $\Delta_5(s)$ in the range of
$t$ from zero to 60. The most striking features of these plots are provided by the first quadrant regions which extend rightwards from the region of $\sigma$ shown to  $\sigma=\infty$, and the
fourth quadrant regions which extend from $\sigma=-\infty$ to $\sigma=\infty$ . As we  have proved in Theorem \ref{thm5-1}, the critical line lies in either the fourth quadrant or the second quadrant,
and changes from one to the other at zeros or poles of $\Delta_5(s)$. The separation between the poles and zeros is a highly variable quantity, as exemplified by the example shown in Fig. \ref{fig02} (right panel), which is not resolved in the coarser scale plot of Fig. \ref{fig01} (right panel).

The phase plots of Figs. \ref{fig01} and \ref{fig02} are complemented by the amplitude plots of Fig \ref{fig03}. These have as their most prominent feature the boundaries between regions with
$|\Delta_5(s)|>1$ and $|\Delta_5(s)|<1$, which are approximately symmetric under reflection in the
critical line if $s>>1$, by (\ref{del5-1}) and (\ref{del5-6}). These amplitude boundary lines of course
must pass between poles and zeros, and intersect lines of constant phase orthogonally.

\begin{figure}[h]
\includegraphics[width=3.0in]{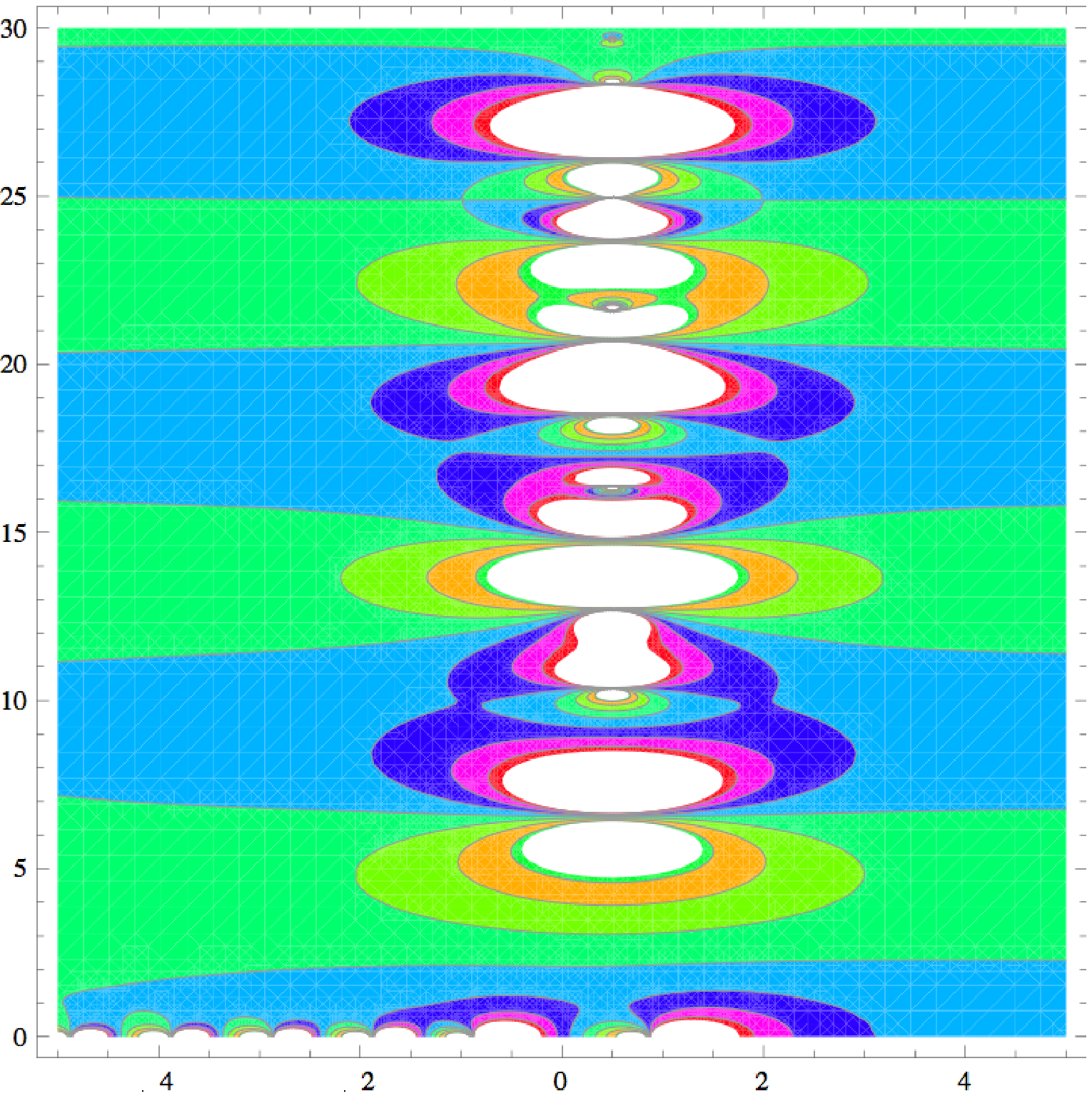}~~\includegraphics[width=3.0in]{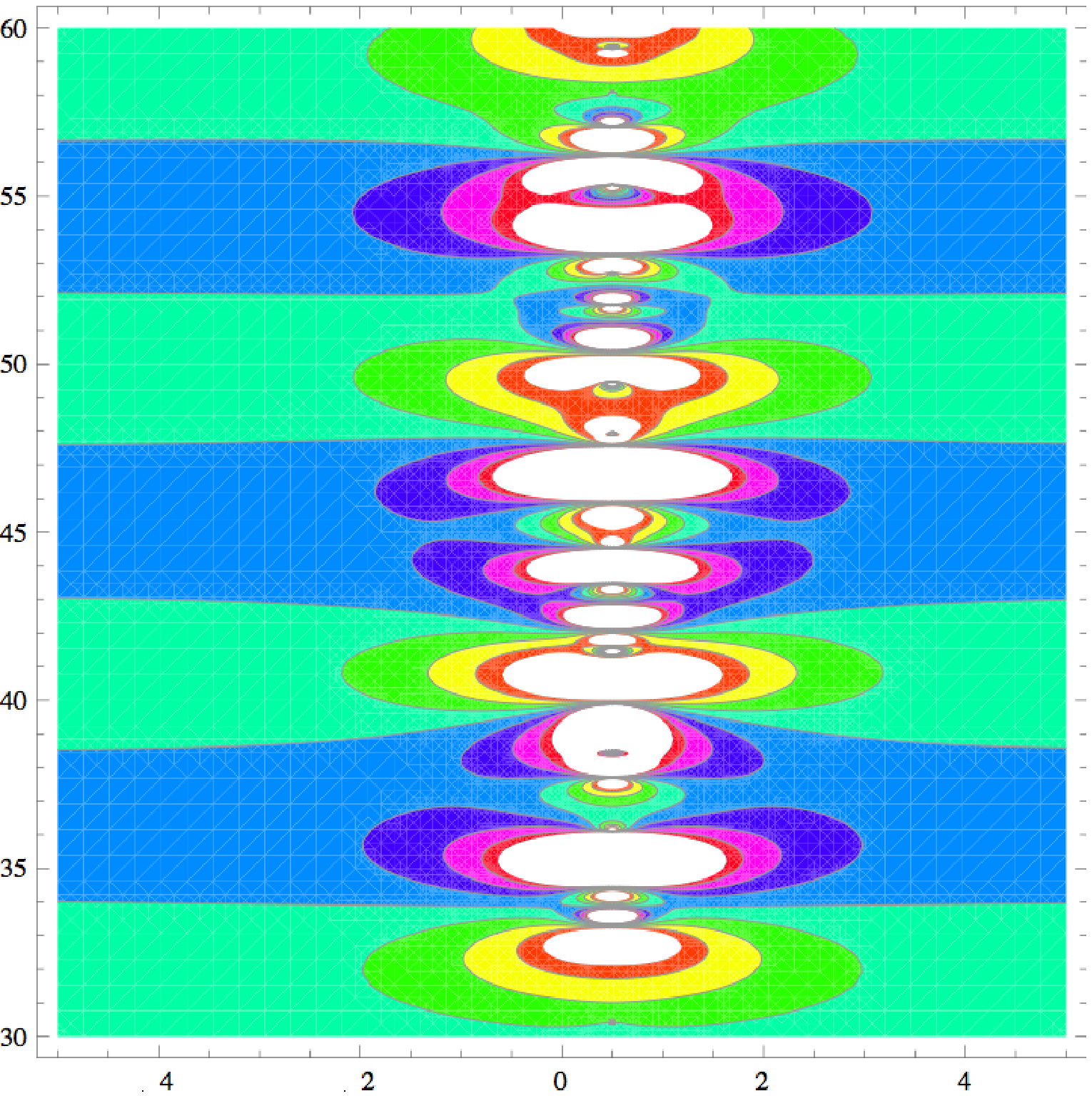}
\caption{\label{fig03}Plots of the amplitude of $\Delta_5(s)$ in the complex plane. The amplitude is denoted by colour, with the blue-toned regions containing contours of amplitude increasing from 1 towards a pole at their centre, and the green-toned regions containing contours of amplitude decreasing from 1 towards a zero at their centre.
.}
 \end{figure}
\begin{theorem}
The only lines of constant phase of $\Delta_5(s)$ which can attain $\sigma=\infty$ are
equally spaced, and have interspersed lines of constant modulus. All such lines of constant phase
reach the critical line at a pole or zero of  $\Delta_5(s)$.
\label{thm5-2}
\end{theorem}
\begin{proof}
The proposition is true for $t$ not large compared with unity on the basis of numerical evidence
(see Figs. \ref{fig01}-\ref{fig03}).

We begin by taking the leading terms in equation (\ref{del5-4}) for $\sigma>>1$:
\begin{equation}
\Delta_5(\sigma+ i t)\simeq 1+2^{-\sigma}\cos(t \ln 2)-i 2^{-\sigma}\sin(t \ln 2).
\label{del5-9}
\end{equation}
From (\ref{del5-9}) the leading terms in the expansions for the modulus and phase of
$\Delta_5(s)$ are
\begin{equation}
\arg[\Delta_5(\sigma+ i t)]\simeq -\frac{\sin( t\ln 2)}{2^{\sigma}},~~
|\Delta_5(\sigma+ i t)|\simeq 1+\frac{\cos( t\ln 2)}{2^{\sigma}}  .
\label{del5-10}
\end{equation}
Thus, the lines of constant  phase zero are those which can attain $\sigma=\infty$, and correspond to $t=n\pi/\ln 2$, for $n$ integral. In between these are lines of constant amplitude unity, which also can attain $\sigma=\infty$, and correspond to $t=(n+1/2)\pi/\ln 2$, for $n$ integral.

For large $\sigma$, the lines of phase zero mark the boundary between regions in which the
phase of $\Delta_5(s)$ lies in the fourth quadrant (on one side of the line) and the first quadrant (on the other side).
The amplitude of  $\Delta_5(s)$  in this region is given by $1+(-1)^n/2^{\sigma}$, and thus
decreases as $\sigma$ decreases for $n$ odd, and increases for $n$ even. The amplitude is a monotonic function of $\sigma$, and must continue to behave monotonically even when
$\sigma$ is insufficiently large for (\ref{del5-9}) and (\ref{del5-10}) to be accurate. Indeed, if say
the amplitude were decreasing, and then started to increase, the lines of constant amplitude forming around the central line of constant phase would have to form closed loops, not possible
unless zeros and poles of $\Delta_5(s)$ lay within the loops. A similar argument of course applies
to the case of increasing amplitude along the line of zero phase. Lines of phase zero cannot cut the critical line other than at a zero or pole of  $\Delta_5(s)$ by virtue of its phase  value
(\ref{del5-8}) on the critical line.

Turning now to the lines of constant amplitude of unity, on one side the amplitude exceeds unity, and on the other it  is less than unity. The phase of $\Delta_5(s)$ along these lines in the region of large $\sigma$ is given by $-(-1)^n /2^{\sigma}$, and so lies in the fourth quadrant for $n$ even, and in the first for $n$ odd. It increases monotonically in its magnitude as $\sigma$ decreases, and again this conclusion holds true even if $\sigma$ is insufficiently large for (\ref{del5-9}) and (\ref{del5-10}) to be accurate.  Lines of constant unit amplitude  must cut the critical line in between poles and zeros
of $\Delta_5(s)$ .

We next consider the properties of $\Delta_5(s)$ in $\sigma<1/2$. Combining (\ref{del5-6}) and (\ref{del5-9}) the functional equation for $\Delta_5(s)$ gives
\begin{equation}
\Delta_5(1-\sigma+i t)\simeq \left[1+2^{-\sigma} \cos(t\ln 2)+i 2^{-\sigma} \sin(t\ln 2)\right]e^{-i\pi/4}
 \left( 1-\frac{\sigma+i t}{16(\sigma^2+t^2)}+\ldots \right),
\label{del5-11}
\end{equation}
a good approximation in $\sigma>3$. Comparing (\ref{del5-11}) and (\ref{del5-9}), we see
that the modulus is even under the transition from $\Delta_5(\sigma+i t)$ to $\Delta_5(1-\sigma+i t)$.  For the phase, the imaginary $t$-dependent part has the opposite sign and there is a translation of $-\pi/4$,
resulting in  $\Delta_5(1-\sigma+i t)$ everywhere in $t>1$ and $\sigma>3$ having a phase in the fourth quadrant.

Consider finally the possibility that the lines of constant amplitude and phase coming from $\sigma=\infty$ ceased
for some value of $t$ to reach the critical line. For this to occur, the lines of constant phase and amplitude would have to curve upwards and run towards $t=\infty$, and they would have to do that
everywhere above the starting value of $t$. (Otherwise, they would cut subsequent regions of lines of constant phase and amplitude heading towards the critical line, requiring the presence of accumulation points of zeros and poles in the finite part of the plane, not in keeping with the properties of ${\cal C}(0,1;s)$ and $\zeta (2s -1/2)$.) However, such lines coming from $\sigma=\infty$ would require, from the functional equation  corresponding lines to exist mirrored about $\sigma=1/2$. Such lines on the right would have phases in the first and fourth quadrants, and on the left their phase values for $|s|>>1$ would lie in the interval $[-3\pi/4,\pi/4]$.
The lines of constant amplitude are symmetric under $s\rightarrow 1-s$ if $|s|>>1$. This means sets of lines of constant amplitude are required to start in $\sigma<1/2$ at high $t$ values and proceed as $\sigma$ decreases towards smaller values of $t$. However, from equation  (\ref{del5-11}), as $\sigma$ increases, the algebraic terms in the parentheses dominate the values of $|\Delta_5(1-\sigma+i t)|$, with the exponentially-decreasing terms having a much smaller effect. It is easy to show that the contours of constant amplitude $A$ of the term in parentheses in (\ref{del5-11}) are in fact circles
given by:
\begin{equation}
\left(\sigma -\frac{1}{16 (1-A^2)}\right)^2+t^2=\frac{A}{16|1-A^2|}.
\end{equation}
These circles lie in the required region $\sigma$ significantly larger than one only for amplitudes $A$ close to but less than unity. This means that the contours of constant amplitude of  $\Delta_5(1-\sigma+i t)$ which can proceed from $1-\sigma<1/2$ towards strongly negative values are limited to amplitudes close to but less than unity. This gives a contradiction, since the lines of constant amplitude curving up
towards the critical line correspond to amplitudes both larger and smaller than unity. In summary, lines of constant amplitude and phase curving up towards infinite $t$ near the critical line are ruled out by the asymptotic behaviour of $\Delta_5(1-\sigma+i t)$ for large $\sigma$.


\end{proof}
\begin{corollary}
The only lines of constant phase of $\Delta_5(s)$ which can attain $\sigma=-\infty$ are
equally spaced, and have interspersed lines of constant modulus. All such lines of constant phase
reach the critical line at a pole or zero of  $\Delta_5(s)$.
\label{corol5-4}
\end{corollary}
\begin{proof}
The proof of this proposition is contained in the discussion of Theorem \ref{thm5-2}, and specifically in the  discussion of equation (\ref{del5-11}), the functional equation (\ref{del5-1}) and
equation (\ref{del5-6}). The lines of constant phase referred to in fact correspond to the phase $-\pi/4$. There is no ambiguity (modulo $2\pi$) in this phase in $\sigma<<1/2$, where the phase is controlled by the last two terms in (\ref{del5-11}), and since the fourth quadrant region extends continuously from $\sigma\rightarrow-\infty$ to $\sigma\rightarrow +\infty$, the phase of $\Delta_5(s)$ is
similarly unique in $\sigma>1/2$.
\end{proof}
\section{The Riemann Hypothesis for $\zeta(s)$ and $L_{-4}(s)$}

Using Theorem \ref{thm5-2}, we are in a position to follow the argument of I, to prove the equivalent of its  Theorem 4.3.
\begin{theorem}
The Riemann zeta function obeys the Riemann hypothesis if and only if the Dirichlet beta function
obeys the Riemann hypothesis.
\label{thm5-3}
\end{theorem}
\begin{proof}
Consider lines $L_1$ and $L_2$  in $t>0$ along which the phase of $\Delta_5(s)$  is zero. Join these
lines with two lines  to the right of the critical line along which  $\sigma$ is constant. Then the change of argument of  $\Delta_5(s)$ around the closed contour  $C$ so formed is zero, so by the Argument Principle the number of poles inside the contour equals the number of zeros. Each non-trivial pole is formed by a zero of $\zeta(s)$, so if there are no such zeros within $C$ there can be no zeros of  ${\cal C}(0,1;s)$  and so of $L_{-4}(s)$ within $C$. 

Conversely,  suppose there are no zeros of $L_{-4} (s)$ off the critical line.  Suppose the
first zero of $\zeta(s)$ off the critical line occurs at $t=t_0$. Then there is a zero of $\zeta(2 s-1/2)$ off the critical line at $t=t_0/2$, and so a pole of $\Delta_5(s)$ off the critical line there, with no possible zero off the critical line in that neighbourhood of $t$.This contradicts the Argument Principle. 

These arguments prove the theorem in the region to the right of the critical line lying between lines of zero phase of $\Delta_5(s)$, with the result to the left of the critical line then guaranteed by the functional equation (\ref{del5-1}) .

To complete the proof we need to show that any point in the region $\sigma>1/2$, $t>0$ is enclosed between lines of phase zero of  $\Delta_5(s)$ coming from $\sigma=\infty$. We note that $t=0$ is one such line, and that for any $\sigma>1/2$ the infinite number of such constant phase lines cannot cluster into a finite interval of $t$, since that would indicate an essential singularity of  $\Delta_5(s)$ for that $\sigma$.
\end{proof}

\begin{corollary}
Suppose it is known that the Riemann hypothesis holds for $\zeta(s)$ in the range $t<t_0$.
Then the Riemann hypothesis holds for $L_{-4}(s)$ at least  in the range $t<t_0/2$.
\label{corol5-1}
\end{corollary}
\begin{proof}
We know the first off-axis zero of $\zeta(s)$ has to occur for some $t_1>t_0$. Then $\Delta_5(s)$
has to have a pole off-axis at $t=t_1/2$, and so it must have an off-axis zero between two enclosing contours
of zero phase of $\Delta_5(s)$ coming from $\sigma=\infty$. This must correspond to a zero
of $L_{-4}(s)$, and there can be no previous off-axis zero between a lower pair of contours of zero phase.
\end{proof}

Combining theorem \ref{thm5-3} with theorem 4.3 of I, we have immediately
\begin{corollary} If either $\zeta (s)$ or $L_{-4}(s)$ obeys the Riemann hypothesis, then so does the other, and all the two-dimensional sums ${\cal C}(1, 4 m;s)$ for non-negative integers  $m$.
\label{corol5-2}
\end{corollary}
\section{Distribution Functions for the Zeros of $\zeta(s)$ and $L_{-4}(s)$}
\begin{theorem}
Assuming the Riemann hypothesis holds for $\zeta(s)$ or $L_{-4}(s)$, then given any two lines of phase zero of $\Delta_5(s)$ running from $\sigma=\infty$ and intersecting the critical line, the number of zeros and poles of $\Delta_5(s)$ counted according to multiplicity and lying properly between the lines must be the same.
\label{thm5-4}
\end{theorem}
\begin{proof}
We follow the argument in Theorem 4.5 of McPhedran et al (2013). We consider  a contour composed of two lines of phase zero, the segment between them on the critical line, and a segment between them
on the interval $\sigma>>1$. The total phase change  around this contour is strictly zero, since the region $\sigma>>1$ has the phase of $\Delta_5(\sigma+i t)$ constrained to be close to zero. Let $P_u=(1/2,t_u)$ be the point at the upper end of the segment on the critical line, and $P_l=(1/2,t_l)$ be the point at the lower end.

The total phase change between a point approaching $P_u$ on the contour from the right and a point leaving $P_l$ going right is, by construction, zero. This phase change is made up of contributions from the changes of phase at the zero or pole $P_u$, from the zero or pole $P_l$, from the $N_z$ zeros and $N_p$ poles on the critical line between $P_u$ and $P_l$, and from the phase change between the zeros and poles. In this list, the first phase change is $\phi_5(t_u)$, the phase of $\Delta_{5}(s)$ on the critical line just below $t_u$. The second change is $-\phi_5(t_l)$, where $\phi_5(t_l)$ is the phase of 
$\Delta_5(s)$  just above $t_l$. Giving zero $n$ a multiplicity $z_n$, and pole $n$ a multiplicity $p_n$, the phase change at the former is $-\pi z_n$ and at the latter $\pi p_n$. The phase change between zeros and poles is $\phi_5(t_l)-\phi_5(t_u)$. Hence the phase change is
\begin{equation}
\phi_5(t_u)-\phi_5(t_l)-\pi \left(\sum_{n=1}^{N_z}z_n-\sum_{n=1}^{N_p}p_n\right)+\phi_5(t_l)-\phi_5(t_u)=0,
\label{nthm5-4a}
\end{equation}
leading to
\begin{equation}
\sum_{n=1}^{N_z}z_n=\sum_{n=1}^{N_p}p_n ,
\label{nthm5-4b}
\end{equation}
as asserted.

\end{proof}
\begin{corollary}
Assuming the Riemann hypothesis holds for $\zeta(s)$ or $L_{-4}(s)$, if all zeros and poles on the critical line have multiplicity one, the numbers of zeros and poles on the critical line between any pair of lines of phase zero of $\Delta_5(s)$ coming from $\sigma=\infty$ are the same. The distribution functions for zeros of ${\cal C}(0,1;s)$ and $\zeta (2s -1/2)$
on the critical line
must then agree in all terms which   do not go to zero as $t\rightarrow \infty$.
\label{corol5-3}
\end{corollary}
\begin{proof}
The first assertion is a simple consequence of theorem \ref{thm5-4}. The second follows from
theorems \ref{thm5-2} and \ref{thm5-4}: the number of zeros and poles between successive zero lines coming from $\sigma=\infty$ match for all such pairs of lines. There are no zeros or poles on the critical line before the first line of phase zero off the real axis coming from $\sigma=\infty$,
so the distribution functions of zeros of the numerator and denominator of $\Delta_5(s)$
must then agree exactly at  the infinite set of values of $t$ corresponding to the intersection points above this first line, and in neighbourhoods including each member of the set.  This precludes
those distribution functions differing by terms which do not go to zero as $t\rightarrow \infty$.
\end{proof}

From Titchmarsh and Heath-Brown (1987), given the Riemann hypothesis holds for $\zeta(s)$,
the distribution function for its zeros on the critical line is
\begin{equation}
N_\zeta(\frac{1}{2},t)=\frac{t}{2\pi} \log (t)-\frac{t}{2\pi}(1+\log(2\pi))+I_\zeta(\frac{1}{2},t)+F_\zeta(\frac{1}{2},t),
\label{del5-12}
\end{equation}
where the functions $I$ and $F$ denote the sums of all terms which go to infinity as $t\rightarrow \infty$ or which remain finite, respectively. Hence, we have
\begin{equation}
N_\zeta(\frac{1}{2},2 t)=\frac{t}{\pi} \log (t)-\frac{t}{\pi}(1+\log(\pi))+I_\zeta(\frac{1}{2},2 t)+F_\zeta(\frac{1}{2},2 t),
\label{del5-13}
\end{equation}
From the result of corollary \ref{corol5-3}, we have for the distribution function of the zeros of $L_{-4}(s)=\beta(s)$ on the critical line that
\begin{equation}
N_{\beta}(\frac{1}{2},t)=\frac{t}{2\pi} \log (t)-\frac{t}{2\pi}(1+\log(\pi/2))+I_{\beta}(\frac{1}{2},t)+F_{\beta}(\frac{1}{2},t),
\label{del5-14}
\end{equation}
where 
\begin{equation}
I_\zeta(\frac{1}{2},2 t)=I_\zeta(\frac{1}{2}, t)+I_{\beta}(\frac{1}{2},t).
\label{del5-15}
\end{equation}
Thus, the distribution function for the zeros of $\zeta(s)$ determines that for $L_{-4}(s)$, in all important terms.

\section{Comments on Higher Order Dirichlet $L$ Functions}
Suppose we are interested in investigating the Riemann hypothesis for other Dirichlet $L$ functions $L_{-q}(s)$. From McPhedran {\em et al} (2008) we have for the distribution function for zeros on the critical line of  $L_{-q}(s)$,
\begin{equation}
N_{L-q}(\frac{1}{2},t)=\frac{t}{2\pi} \log t-\frac{t}{2\pi} (1+\log (2\pi/q))+\ldots .
\label{del5-16}
\end{equation}

We will want to construct a ratio of functions whose numerator and denominator have the same distribution functions of zeros on the critical line, as far as the terms given in equations (\ref{del5-12},\ref{del5-13},\ref{del5-16}) are concerned. For the case $q=3$, the product
$\zeta(s) L_{-3}(s)$ will have a lower density of  zeros on the critical line than
$\zeta(2 s-1/2)$, and thus extra zeros need to be inserted  in the numerator using an exponential factor. We define
\begin{equation}
\Delta_{-3}(s)=\frac{[1-(\frac{3}{4})^{s-1/2}] \zeta(s) L_{-3}(s)}{\zeta(2 s-1/2)}.
\label{del5-17}
\end{equation}
The term in square brackets has been chosen so that it tends to unity as $\sigma\rightarrow \infty$.
Its lines of zero phase there are given by $t=m\pi/log(4/3)$ for positive integers $m$.

For $q>4$, the extra zeros are required in the denominator. We define
\begin{equation}
\Delta_{-q}(s)=\frac{ \zeta(s) L_{-q}(s)}{[1-(\frac{4}{q})^{s-1/2}]\zeta(2 s-1/2)}.
\label{del5-18}
\end{equation}
The lines of zero phase for large positive $\sigma$ are given by $t=m\pi/log(q/4)$ for positive integers $m$.

These suggestions are of course extrapolations based on the case $q=4$, and need to be investigated numerically and analytically.

The work of R.C. McPhedran on this project has been supported by the Australian Research Council's Discovery Grants Scheme. He also acknowledges the financial support of the European Community's Seventh Framework
Programme under contract number PIAPP-GA-284544-PARM-2. M

\end{document}